\documentclass[a4paper,UKenglish]{lipics-v2018}

\usepackage[subtle]{savetrees}
\usepackage{graphicx}
\usepackage{caption}
\usepackage[skip=-3pt,belowskip=4pt]{subcaption}
\usepackage{xspace} 
\usepackage{textcomp}
\usepackage{balance}  
\usepackage{xcolor}
\usepackage{wrapfig}
\usepackage{amsmath}
\usepackage{mathtools}
\usepackage{algorithm}
\usepackage{algorithmicx}
\usepackage{cite}
\usepackage{float}
\usepackage{tablefootnote}
\usepackage{enumitem} 
\usepackage[noend]{algpseudocode}
\usepackage{url}
\usepackage{siunitx}
\usepackage{amssymb}
\usepackage{color}
\usepackage{booktabs}
\usepackage{xspace}
\usepackage{caption}
\usepackage{subcaption}
\usepackage{cleveref}

\usepackage{wrapfig}

\nolinenumbers

\usepackage{thmtools, thm-restate}

\usepackage{microtype}

\newif\ifnotes
\notesfalse
\notestrue
\ifnotes
\newcommand{\dzejla}[1]{\textcolor{blue}{{\footnotesize #1}}\marginpar{\raggedright\tiny \textcolor{blue}{Dzejla}}}
\definecolor{purple}{rgb}{0.5,0,0.5}
\newcommand{\mayank}[1]{\textcolor{purple}{{\footnotesize #1}}\marginpar{\raggedright\tiny \textcolor{purple}{Mayank}}}
\definecolor{green}{rgb}{0,1,0}
\newcommand{\emina}[1]{\textcolor{green}{{\footnotesize #1}}\marginpar{\raggedright\tiny \textcolor{orange4}{Emina}}}
\definecolor{orange}{rgb}{1,0.6,0}
\newcommand{\ehsan}[1]{\textcolor{orange}{{\footnotesize #1}}\marginpar{\raggedright\tiny \textcolor{orange}{Ehsan}}}
\newcommand{\prashant}[1]{\textcolor{red}{{\footnotesize #1}}\marginpar{\raggedright\tiny \textcolor{red}{Prashant}}}

\renewcommand{\dzejla}[1]{{\ifnotes \scriptsize \textcolor{green}{Dzejla: {#1}} \fi}}
\renewcommand{\mayank}[1]{\ifnotes {\noindent \scriptsize  \textcolor{blue} {Mayank: {#1}}} \fi{}}
\renewcommand{\emina}[1]{{\ifnotes \scriptsize \textcolor{orange}{Emina: {#1}} \fi}}
\renewcommand{\ehsan}[1]{{\ifnotes \scriptsize \textcolor{orange}{Ehsan: {#1}} \fi}}
\renewcommand{\prashant}[1]{{\ifnotes \scriptsize \textcolor{red}{Prashant: {#1}} \fi}}

\else
\newcommand{\dzejla}[1]{}
\newcommand{\mayank}[1]{}
\newcommand{\emina}[1]{}
\newcommand{\ehsan}[1]{}
\newcommand{\prashant}[1]{}

\fi

\newcommand{\bcms}{buffered count-min sketch\xspace}

\newcommand{\cms}{count-min sketch\xspace}

\title{Buffered Count-Min Sketch on SSD: Theory and Experiments}

\author{Mayank Goswami}{Queens College, City University of New York{}}{mayank.goswami@qc.cuny.edu}{}{}

\author{Dzejla Medjedovic}{International University of Sarajevo{}}{dzmedjedovic@ius.edu.ba}{}{}

\author{Emina Mekic}{Sarajevo School of Science and Technology{}}{emina.mekic@stu.ssst.edu.ba}{}{}

\author{Prashant Pandey}{Stony Brook University, New York{}}{ppandey@cs.stonybrook.edu}{}{}
\authorrunning{M.\ Goswami, D.\ Medjedovic, E.\ Mekic, and P.\ Pandey}

\Copyright{Mayank Goswami and Dzejla Medjedovic and Emina Mekic and Prashant
Pandey}

\subjclass{
\ccsdesc[500]{Theory of computation~Data structures and algorithms for data management},
\ccsdesc[300]{Theory of computation~Streaming models},
\ccsdesc[300]{Theory of computation~Database query processing and optimization (theory),}}

\keywords{Streaming model, Count-min sketch, Counting, Frequency, External memory, I/O efficiency, Bloom filter, Counting filter, Quotient filter.}

\EventEditors{}
\EventLongTitle{European Symposium on Algorithms (ESA 2018)}
\EventShortTitle{ESA 2018}
\EventAcronym{ESA}
\EventYear{2018}
\EventDate{August 20-22, 2018}
\EventLocation{Helsinki, Finland}
\ArticleNo{}

\begin{document}

\maketitle

\begin{abstract}

Frequency estimation data structures such as the count-min sketch (CMS) have
found numerous applications in databases, networking, computational biology and
other domains.
Many applications that use the count-min sketch process massive and rapidly
evolving datasets. For data-intensive applications that aim to keep the
overestimate error low, the count-min sketch may become too large to store in
available RAM and may have to migrate to external storage (e.g., SSD.) Due to
the random-read/write nature of hash operations of the count-min sketch, simply
placing it on SSD stifles the performance of time-critical applications,
requiring about $4$-$6$ random reads/writes to SSD per estimate (lookup) and
update (insert) operation.

In this paper, we expand on the preliminary idea of the Buffered Count-Min
Sketch (BCMS)~\cite{EydiMeMe18}, an SSD variant of the count-min sketch, that
used hash localization to scale efficiently out of RAM while keeping the total
error bounded.
We describe the design and implementation of the buffered count-min sketch, and
empirically show that our implementation achieves $3.7\times$-$4.7\times$ the
speedup on update (insert) and $4.3\times$ speedup on estimate (lookup)
operations.

Our design also offers an asymptotic improvement in the external-memory
model~\cite{AggarwalVi88} over the original data structure: $r$ random I/Os are
reduced to $1$ I/O for the estimate operation. For a data structure that uses
$k$ blocks on SSD, $w$ as the word/counter size, $r$ as the number of rows, $M$
as the number of bits in the main memory, our data structure uses $kwr/M$
amortized I/Os  for updates, or, if $kwr/M > 1$, 1 I/O in the worst case. In
typical scenarios, $kwr/M$ is much smaller than $1$. This is in contrast to
$O(r)$ I/Os incurred for each update in the original data structure.

Lastly, we mathematically show that for the buffered count-min sketch, the error
rate does not substantially degrade over the original count-min sketch due to
hash localization. Specifically, we prove that for any query $q$, our data
structure provides the guarantee: $\text{Pr}[\text{Error}(q) \geq n \epsilon
(1+o(1))] \leq \delta + o(1)$, which, up to $o(1)$ terms, is the same guarantee
as that of a count-min sketch.




\end{abstract}

\section{Introduction}

Applications that generate and process \emph{massive data streams} are becoming
pervasive~\cite{BabcockBaDa02, MankuMo02, NathGiSe04, GaberzaKr05, Woodruff16}
across many domain in computer science.  Common examples of streaming datasets
include financial markets, telecommunications, IP traffic, sensor networks,
textual data, etc~\cite{BabcockBaDa02, ChenCaJi10, ZhangSiSe10, BreslauCaFa99}.
Processing fast-evolving and massive datasets poses a challenge to traditional
database systems, where commonly the application stores all data and
subsequently does queries on it. In the streaming
model~\cite{babcock2002models}, the dataset is too large to be completely stored
in the available memory, so every data item is seen and processed once --- an
algorithm in this model performs only one scan of data, and uses sublinear local
space.

The streaming scenario exhibits some limitations on the types of problems we can
solve with such strict time and space constraints. A classic example is the
heavy hitter problem \texttt{HH(k)} on the stream of pairs $(a_{t}, c_{t})$,
where $a_{t}$  is the item identifier, and $c_{t}$ is the count of the item at
timeslot $t$, with the goal of reporting all items whose frequency is at least
$n/k$, $n= \sum_{t=1}^{T}c_{t}$. The general version of the problem, with the
exception of when $k$ is a small constant\footnote{When $k \approx 2$ this
problem goes by the name of majority element.}, can not be exactly solved in the
streaming model~\cite{RoughgardenVa18, ZhangSiSe10}, but the approximate version
of the problem, $\epsilon$\texttt{-HH(k)}, where all items of the frequency at
least $n/k-\epsilon n$ are reported, and an item with larger error might be
reported with small probability $\delta$, is efficiently solved with the
count-min sketch~\cite{CormodeMu05, MuthukrishnanCo11} data structure. Count-min
sketch  accomplishes this in $O(\ln(1/\delta)/\epsilon)$ space, usually far
below linear space in most applications. 

Count-min sketch~\cite{CormodeMu05, MuthukrishnanCo11} has been extensively used
to answer heavy hitters, top $k$ queries and other \emph{popularity measure
queries} that represent the central problem in the streaming context, where we
are interested in extracting the essence from an impractically large amount of
data. Common applications include displaying the list of bestselling items, the
most clicked-on websites, the hottest queries on the search engine, most
frequently occurring words in a large text, and so on~\cite{SchechterHeMi10,
NathGiSe04, ZhaoOgWa06}. 

%
%

Count-min sketch (CMS) is a hashing-based, probabilistic and lossy
representation of a multiset, that is used to answer the count of a query $q$
(number of times $q$ appears in a stream). It has two error parameters: 1)
$\epsilon$, which controls the overestimation error, and 2) $\delta$, which
controls the failure probability of the algorithm. The CMS provides the
guarantee that the estimation error for any query $q$ is more than $\epsilon n$
with probability at most $\delta$. If we set $r = \ln(1/\delta)$ and $c =
e/\epsilon$, the CMS is implemented using $r$ hash functions as a  2D array of
dimensions $r$ x $c$.

When $\epsilon$ and $\delta$ are constants, the total overestimate grows
proportionately with $n$, the size of the count-min sketch remains small, and
the data structure easily fits in smaller and faster levels of memory. For some
applications, however, the allowed estimation error of $\epsilon n$ is too high
when $\epsilon$ is fixed. Consider an example of $n=2^{30}$, where $\delta=0.01$
and $\epsilon = 2^{-26}$, hence the overestimate is 16, and the total data
structure size of 3.36GB, provided each counter uses 4 bytes. However, if we
double the dataset size, then the total overestimate also doubles to 32 if
$\epsilon$ stays the same. On the other hand, if we want to maintain the fixed
overestimate of 16, then the data structure size doubles to 6.72GB. 

In this paper, we expand on the preliminary idea of Buffered Count-Min Sketch
(BCMS)~\cite{EydiMeMe18}, an SSD variant of the count-min sketch data structure,
that scales efficiently to large datasets while keeping the total error bounded.
Our work expands on the previous work by introducing detailed design,
implementation and experiments, as well as mathematical analysis of the new data
structure (our original paper~\cite{EydiMeMe18}, which, to the best of our
knowledge is the only attempt thus far to scale count-min sketch to SSD,
contains only the outline of the data structure). 

To demonstrate the issues arising from a growing count-min sketch and storing it
in lower levels of memory, we run a mini in-RAM experiment for count-min sketch
sizes 4KB-64MB. In Figure~\ref{exp0}, we see that to maintain the same error,
the cost of update will increase as the data structure is being stored in the
lower levels of memory, even though we keep the number of hash functions fixed
for all data structure sizes. The appropriate peak in the cost is visible at the
border of L2 and L3 cache (at 3MB). 

\begin{figure}
\centering
\includegraphics[width=10cm]{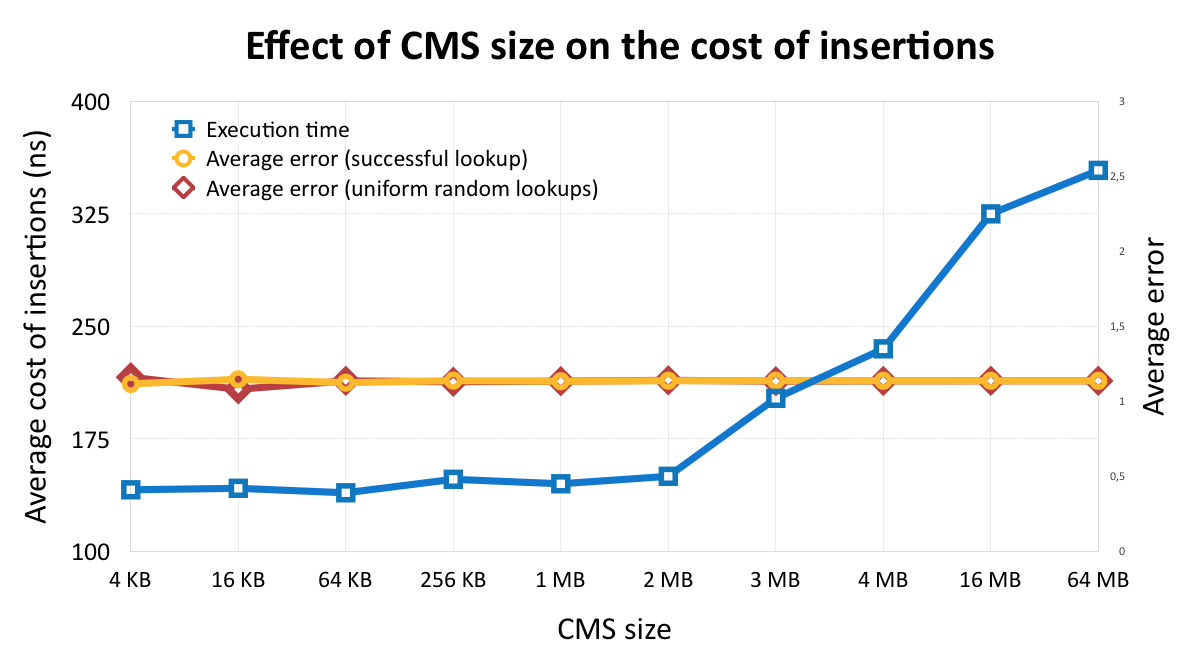}
\caption{The effect of increasing count-min sketch size on the update operation cost in RAM.}
\label{exp0}
\end{figure}

Asymptotically, storing the unmodified count-min sketch on SSD or a disk is
inefficient, given that each estimate and update operation needs $r$ hashes,
which results in $O(r)$ random reads/writes to SSD, far below the desired
throughput for most time-critical streaming applications.

Another context where we see CMS becoming large even when $\epsilon$ is fixed is
in some text applications,where the number of elements inserted in the sketch is
quadratic in the original text size. For instance,~\cite{GoyalJaDa10} uses CMS
to record distributional similarity on the web, where each pair of words is
inserted as a  single item into the CMS, and 90GB of text requires a CMS of 8GB.

\subsection{Results}
\begin{enumerate}
\item We describe the design and implementation of buffered count-min sketch, and empirically show that our implementation achieves 3.7-4.7x the speedup on update (insert) and 4.3x speedup on estimate (lookup) operations.
\item Our design also offers an asymptotic improvement in the external-memory model~\cite{AggarwalVi88} over the original data structure: $O(r)$ random I/Os are reduced to 1 I/O for estimate. For a data structure that uses $k$ blocks on SSD, $w$ as the word/counter size, $r$ as the number of rows, $M$ as the number of bits in main memory, our data structure uses $kwr/M$ amortized  I/Os  for updates, or, if $kwr/M > 1$, 1 I/O in the worst case. In typical scenarios,  $kwr/M < < 1$. This is in contrast to $O(r)$ I/Os incurred for each update in the original data structure.

\item We mathematically show that for buffered count-min sketch, the error rate does not substantially degrade over the original count-min sketch. Specifically, we prove that for any query $q$, our data structure provides the following guarantee:
$$ \text{Pr}[\text{Error}(q) \geq n \epsilon (1+o(1))] \leq \delta + o(1).$$
\end{enumerate}

We focus on scenarios where the allowed estimation error is sublinear in $n$. For example, what if we want the estimation error to be no larger than $n/ \log n$, or $\sqrt{n}$? These scenarios correspond to $\epsilon = 1/\log n$ or $1/\sqrt{n}$, and now for even moderately large values of $n$, the count-min sketch becomes too large to fit in main memory. Even given more modest condition, such as $\epsilon = o(1/M)$, where the memory is of size $M$, the count-min sketch is unlikely to fit in memory. We will assume that $1/n \leq \epsilon << 1/M$. Higher values of $\epsilon$ do not require the count-min sketch to be placed on disk, and lower values of $\epsilon$ mean exact counts are desired.

\section{Related Work}
The streaming model represents many real-life situations where the data is produced rapidly and on a constant basis. Some of the applications include sensor networks~\cite{NathGiSe04}, monitoring web traffic~\cite{SarlosBeCs06}, analyzing text~\cite{GoyalJaDa10}, and monitoring satellites orbiting the Earth~\cite{GertzQuRu06}.

Heavy hitters, top $k$ queries, iceberg queries, and quantiles ~\cite{Woodruff16, NathGiSe04, BabcockBaDa02} are some of the most central problems in the streaming context, where we wish to extract the general trends from a massive dataset. Count-Min sketch has proved useful in such contexts for its space-efficiency and providing accurate counts~\cite{CormodeMu05, MankuMo02}. 

Count-Min sketch can be well illustrated using its connection to the Bloom filter~\cite{Bloom70, BroderMi05, BonomiMiPa06}. Both data structures are lossy and space-efficient representations of sets, used to reduce disk accesses in time-critical applications. Bloom filter answers membership queries and can have false positives, while Count-Min sketch answers frequency queries, and can overestimate the actual frequency count. Both data structures are hashing-based, and suffer from similar issues when placed directly to SSD or a magnetic disk. 

There has been earlier attempts to scale Bloom filters to SSD using buffering and
hash localization~\cite{CanimMiBi10, DebnathSeLi11}. Our paper employs similar
methods to those in~\cite{CanimMiBi10, DebnathSeLi11}. The improvement, both in our case and in the case of Buffered Bloom filter~\cite{CanimMiBi10} is achieved at the expense of having an extra hash function that helps determine to which page each element is going to hash.

There has also been work in designing cache-efficient equivalents for Bloom filters
such as quotient filter and write-optimized on-disk quotient filter such as Cascade filter (CQF)~\cite{BenderFaJo12, DuttaNaBe13, PandeyBJP17}. An important distinction to make between these data structures and count-min sketch is that CQF gives exact counts of most of the elements given that the errors caused by false positives are usually very small. However, since the errors are independent, the CQF doesn't offer any guarantees on the overestimate. For example, two highly occurring elements in a multi-set can collide with each other and both will have large overcounts. On the other hand, the CMS does not give exact counts of elements due to multiple hashes and its size (width of the CMS is smaller than the number of slots in a CQF). But the CMS can offer a guarantee that overestimate will be smaller than $\epsilon n$ with a probability of $\delta$.

\subsection{Count-Min Sketch: Preliminaries}
In the streaming model, we are given a stream $A$ of pairs $(a_{i},
c_{i})$, where $a_{i}$ denotes the item identifier (e.g., IP address, stock ID, product ID), and $c_{i}$ denotes the count of the item (e.g., the number of bytes sent from the IP address, the amount by which a stock has risen/fallen or the number of sold items). 
Each pair $X_{i} = (a_{i}, c_{i})$ is an item within a stream of length $T$, and the goal is to record total sum of frequencies for each particular item $a_{i}$.

For a given estimation error rate $\epsilon$ and failure probabiltity $\delta$, define $r = \ln (1/\delta)$ and $c = e/\epsilon$. The Count-Min Sketch is represented via 2D table with $c$ buckets
(columns), $r$ rows, implemented using $r$ hash functions (one hash function per
row). 

CMS has two operations: \texttt{UPDATE($a_i$)} and \texttt{ESTIMATE($a_{i}$)}, the
respective equivalents of \texttt{insert} and \texttt{lookup}, and they are
performed as follows: 
%

\begin{itemize}
\item \texttt{UPDATE($a_{i}$)} inserts the pair by computing $r$ hash functions on $a_{i}$, and incrementing appropriate slots determined by the hashes by the quantity $c_{i}$. That is, for each hash function $h_{j}$, $1 \le j \le r$, we set $CMS[j][h_{j}(a_{i})]= CMS[j][h_{j}(a_{i})] + c_{i}$. Note that in this paper, we use $c_{i} = 1$, so every time an item is updated, it is just incremented by 1.

\item \texttt{ESTIMATE($a_{i}$)} reports the frequency of $a_{i}$ which can be an overestimate of the true frequency. It does so by calculating $r$ hashes and taking the minimum of the values found in appropriate cells. In other words, we return $min_{1 \le j \le r}(CMS[j][h_{j}(a_{i})])$. Because different elements can hash to the same cells, the count-min sketch can return the overestimated (never underestimated) value of the count, but in order for this to happen, a collision needs to occur in each row. The estimation error is bounded; the data structure guarantees that  for any particular item, the error is within the
range $\epsilon n$, with probability at least $1-\delta$, i.e., $Pr[\text{Error}(q) \geq \epsilon n] \leq \delta$.
\end{itemize}

\section{Buffered Count-Min Sketch}
In this section, we describe Buffered Count-Min Sketch, an adaptation of CMS to SSD.
The traditional CMS, when placed on external storage, exhibits performance issues due to random-write nature of hashing. Each update operation in CMS requires $c=\ln(1/\delta)$ writes to different rows and columns of CMS. On a large data structure, these writes become destined to different pages on disk, causing the update to perform $O(\ln(1/\delta))$ random SSD page writes. For high-precision CMS scenarios where $\delta=0.001\%-0.01\%$, this can be between 5-7 writes to SSD, which is unacceptable in a high-throughput scenario.

To solve this problem, we implement, analyze  and empirically test the data structure presented in~\cite{EydiMeMe18} that outlines three adaptations to the original data structure:
\begin{enumerate}
\item Partitioning CMS into pages and column-first layout: We logically divide the CMS on SSD into pages of block size $B$. CMS with $r$ rows, $c$ columns, cell size $w$, and a total of $S = cr$ $w$-bit counters, contains $k$ pages $P_{1}, P_{2}, P_{3}, \ldots, P_{k}$, where $k=S/B$ and each page spans contiguous $B/r$ columns: $P_{i}$ spans columns $[B(i-1)/r + 1, Bi/r]$. To improve cache-efficiency, CMS is laid out on disk in column-first fashion, which allows each logical page to be laid out sequentially in memory. Thus, each read/write of a logical page requires at most $2$ I/Os.

\item Hash localization: We direct all hashes of each element to a single logical page of CMS that is determined by an additional hash function $h_{0}:[1, k]$. The subsequent $r$ hash functions map to the columns inside the corresponding logical page, i.e., the range of $h_{1}, h_{2}, \ldots, h_{r}$ for an element $e$ is $[B(h_{0}(e)-1)/r + 1, Bh_{0}(e)/r]$. This way,  we  direct all updates and reads related to one element to one logical page.
\item Buffering: When an update operation occurs, the hashes produced for an element are first stored inside an in-memory buffer. The buffer is partitioned into sub-buffers of equal size $S_{1}, S_{2}, \ldots, S_{k}$, and they directly correspond to logical pages on disk in that $S_{i}$ stores the hashes for updates destined for page $P_{i}$. Each element first hashes using $h_{0}$, which determines in which sub-buffer the hashes will be temporarily stored for this element. Once the sub-buffer $S_{i}$ becomes full, we read the page $P_{i}$ from the CMS, apply all updates destined for that page, and write it back to disk. The capacity of a sub-buffer is $M/k$ hashes, which is equivalent to $M/kwr$ elements so the cost of an update becomes $kwr/M << 1$ I/O.
\end{enumerate}


\begin{figure}
\centering
\includegraphics[scale=0.3]{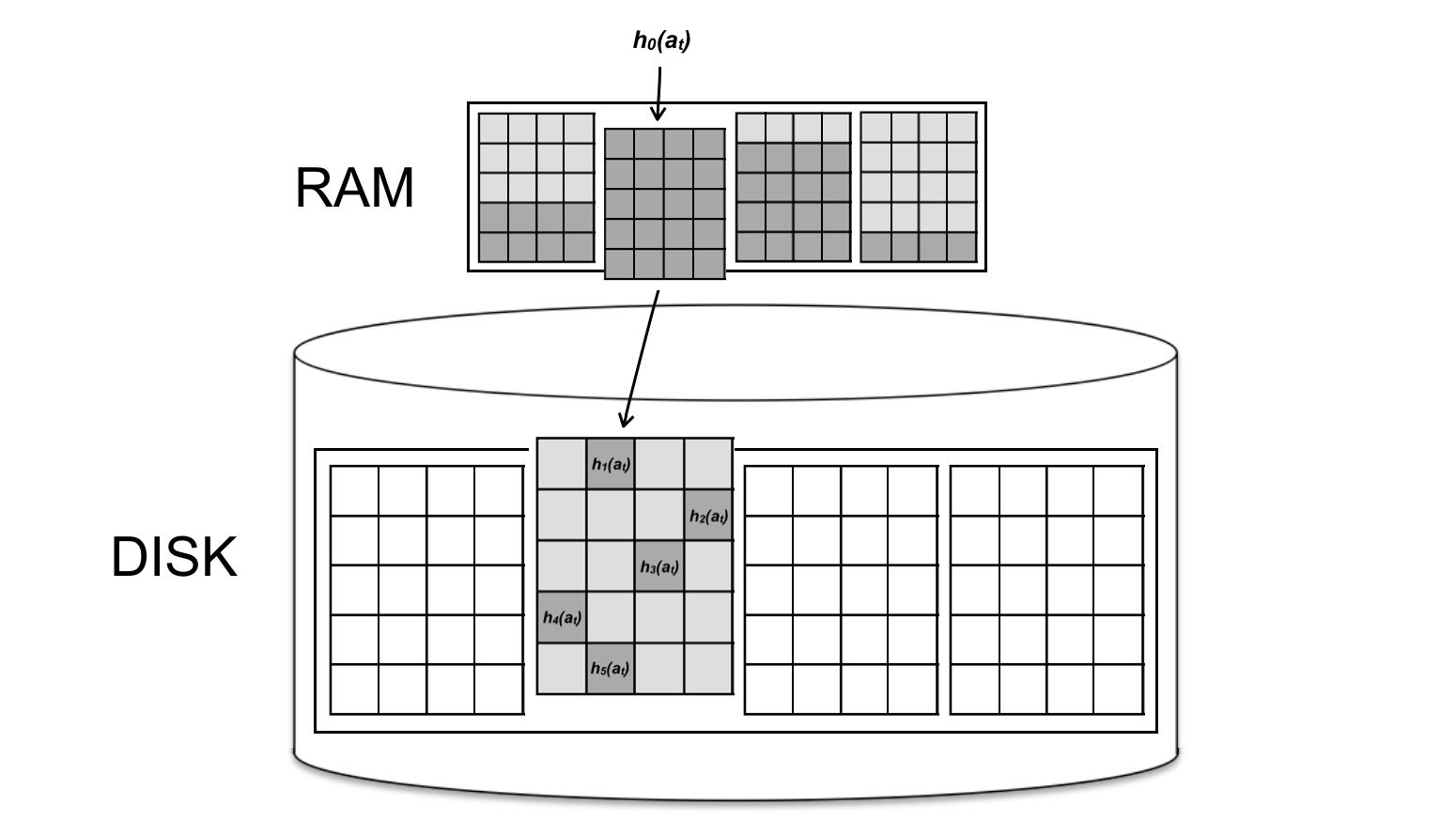}
\caption{\texttt{UPDATE} operation on Buffered Count-Min Sketch. Updates are stored in RAM, and all updates are destined for the same block on disk.}
\label{BCMS}
\end{figure}



 \begin{lstlisting}[caption={Buffered Count-Min Sketch - UPDATE function},label=list1,float,abovecaptionskip=-\medskipamount,mathescape]
 Require: key, r
 subbufferIndex$_i$ :=murmur$_0$(key);
 (*\bfseries for*) i:=1 to r (*\bfseries do*)
    hashes[i] :=murmur$_i$(key);
 (*\bfseries end for*)
 AppendToBuffer(hashes,subbufferIndex);
 
 (*\bfseries if*) isSubbufferFull(subbufferIndex) (*\bfseries then*)
    bcmsBlock :=readDiskPage(subbufferIndex);
    (*\bfseries for each*) entry in Subbuffer[subbufferIndex] (*\bfseries do*)
        (*\bfseries for each*) index in entry (*\bfseries do*)
            pageStart :=calculatePageStart(subbufferIndex);
            offset :=pageStart + entry[index];
            bcmsBlock[offset][index]++;
        (*\bfseries end for*)
    (*\bfseries end for*)
    writeBcmsPageBackToDisk(bcmsBlock);
    clearBuffer(subbufferIndex);
 (*\bfseries end if*)
 \end{lstlisting}

The pseudocode for \texttt{UPDATE($a_{i}$)} is shown in Algorithm~\ref{list1}, and for \texttt{ESTIMATE($a_{i}$)} in Algorithm~\ref{list2}. We use \texttt{murmurhash} as our hashing algorithm due to its efficiency and simplicity~\cite{Appleby11}. Unlike \texttt{UPDATE($a_{i}$)}, \texttt{ESTIMATE($a_{i}$)} operation is not buffered. In a related work~\cite{CanimMiBi10} that implements a buffered Bloom filter on SSD, the data structure buffers lookups. However in the count-min sketch scenario, buffering for \texttt{ESTIMATE($a_{i}$)} is unproductive given that even if the item is found in the buffer, we still need to check the CMS page to obtain the correct count. Therefore, our \texttt{ESTIMATE($a_{i}$)} is optimized for the worst-case single lookup scenario and works for solely insert/lookup as well as mixed workloads. 
The \texttt{ESTIMATE($a_{i}$)} also first computes the correct sub-buffer using $h_{0}$, and flushes the corresponding sub-buffer to SSD page in case some updates were present. Once it applies the necessary changes to the page, it reads the corresponding CMS cells specified by $r$ hashes and returns the minimum value.

 \begin{lstlisting}[caption={Buffered Count-Min Sketch - ESTIMATE function},label=list2,float,abovecaptionskip=-\medskipamount,mathescape]
 Require: key, k
 subbufferIndex$_i$ :=murmur$_0$(key);
 pageStart :=calculatePageStart(subbufferIndex);
 bcmsBlock :=readDiskPage(subbufferIndex);
 
  (*\bfseries if*) isSubbufferNotEmpty(subbufferIndex) (*\bfseries then*)
    (*\bfseries for each*) entry in Subbuffer[subbufferIndex] (*\bfseries do*)
        (*\bfseries for each*) index in entry (*\bfseries do*)
            offset :=pageStart + entry[index];
            bcmsBlock[offset][index]++;
        (*\bfseries end for*)
    (*\bfseries end for*)
    clearBuffer(subbufferIndex);
 (*\bfseries end if*)
 
 (*\bfseries for*) i:=1 to k (*\bfseries do*)
    value :=murmur$_i$(key);
    offset :=pageStart + value;
    estimation :=bcmsBlock[offset][i - 1];
    estimates[i] :=estimation;
 (*\bfseries end for*)
 writeBcmsPageBackToDisk(bcmsBlock);
 (*\bfseries return*) min(estimates)
 \end{lstlisting}

\section{Analysis of Buffered Count-Min Sketch}

In this section, we show that the buffering and hash localization do not substantially degrade the error guarantee of the buffered count-min data structure. Fix a failure probability $ 0 < \delta <1$ and let $0 < \epsilon(n) <1$ be the function of $n$ controlling the estimation error. Let $r = \ln (1/\delta)$ and $c = e/\epsilon$. The traditional count-min sketch uses $S= rc = (e/\epsilon)\ln (1/\delta)$ counters/words of space.  Recall that for our purposes, $1/n \leq \epsilon(n) << 1/M$.


Let $k = S/B$ be the number of blocks occupied by the buffered count-min sketch. We assume a block can hold $B$ counters. Our analysis will assume the following mild conditions:

\noindent\textbf{Assumption 1:}  $n$ is sufficiently larger than the number of blocks $k$, $n =\omega( k(\log k)^{3})$ suffices. Since $k$ depends inversely on $\epsilon(n)$, this assumption essentially means that $\epsilon(n) = \omega(1/n)$.
\noindent\textbf{Assumption 2:} $\lim_{n \rightarrow \infty} \epsilon(n) = 0$. 

Both conditions are satisfied, e.g., when $\epsilon(n) = 1/\log n$ or $1/n^{c}$ for any $c < 1$. 

For brevity, we will drop the dependence of $\epsilon(n)$ on $n$, and write the error rate as just $\epsilon$, however it is important to note that $\epsilon$ is not a constant. 

\begin{theorem}
The Buffered-Count-Min-Sketch is a data structure that uses $k$ blocks of space on disk and for any query $q$,

\begin{itemize}
\item returns \texttt{ESTIMATE}($q$) in 1 I/O and performs updates in $kwr/M$ I/Os amortized, or, if $kwr/M > 1$, in one I/O worst case.
\item Let \texttt{Error}(q) = \texttt{ESTIMATE}(q) - \texttt{TrueFrequency}(q). Then for any $C \geq 1$, \[Pr[\texttt{Error}(q) \geq n\epsilon (1+\sqrt{(2(C+1)k \log k) / n})] \leq \delta + O((\epsilon B/e)^{C}).\]  
\end{itemize}
\end{theorem}

\noindent\textbf{Remark:} By assumption $1$, $\sqrt{(2(C+1)k \log k) / n}$ is $o(1)$  (in fact, it is $o(1/\log k)$). By assumption $2$, $(\epsilon B/e)^{C}$ is $o(1)$. Thus we claim that the buffered count-min-sketch gives almost the same guarantees as a traditional count-min sketch, while obtaining a factor $r$ speedup  in queries.The guarantee for estimates taking 1 I/O is apparent from construction, as only one block needs to be loaded\footnote{In practice, we may need 2 I/Os sometimes due to block-page alignment, but never more than 2}.

The proof is a combination of the classical analysis of CMS and the maximum load of balls in bins when the number of bins is much smaller than the number of balls. Also, note that unlike the traditional CMS, the errors for a query $q$ in different rows are no longer independent (in fact, they are positively correlated: a high error in one row implies more elements were hashed by $h_{0}$ to the same bucket as $q$).

The hash function $h_{0}$ maps into $k$ buckets, each having size $B$ (and so we will also call them blocks). Each bucket can be thought of as a $r \times B/r$ matrix. Note that $r = \ln (1/\delta)$, and $B/r = e/(\epsilon k)$. We assume that $h_{0}$ is a perfectly random hash function, and, abusing notation, identify a bucket/block with a bin, where $h_{0}$ assigns elements (balls) to one of the $k$ buckets (bins).

In this scenario we use Lemma 2(b) from \cite{raab1998balls} and adapt it to our setting.

\begin{lemma}\label{binsnormal}\cite{raab1998balls}
Let $B(n,p)$ denote a Binomial distribution with parameters $n$ and $p$, and $q=1-p$. If $t= np + o((pqn)^{2/3})$ and $x:= \frac{t-np}{\sqrt{pqn}}$ tends to infinity, then 
\[ Pr[B(n,p) \geq t]= e^{-x^{2}/2 - \log x - \frac{1}{2} \log \pi + o(1)}.\]
\end{lemma}

Let $M(n,k)$ denote the maximum number of elements that fall into a bucket, when hashed by $h_{0}$.

\begin{lemma}\label{binsours}
Let $C \geq 1$ and $t= n/k + \sqrt{2(C+1) \frac{n \log k}{k}}$. Then 
\[ Pr[M(n,k) \leq t] \geq 1 - 1/k^C. \]
\end{lemma}

\begin{proof}
We first check that $t$ satisfies the conditions of Lemma~\ref{binsnormal}. Since $h_{0}$ is uniform, $p = 1/k$ (i.e., each bucket is equally probable), and $np = n/k$. We need to check that the extra term in $t$,  $\sqrt{2(C+1)\frac{n \log k}{k}}$ is $o((n(1-1/k)/k)^{2/3})$. This is precisely the condition that $n =\omega( k(\log k)^{3})$ (assumption $1$).

Next we apply Lemma~\ref{binsnormal}. In our case,
\[ x = \sqrt{\frac{2(C+1) n \log k / k}{n (1-1/k) / k}} = \sqrt{2(C+1) \log k (1+ 1/k-1)}, \] 
Now by assumption $2$, $\epsilon(n)$ goes to zero as $n$ goes to infinity, and so $k \propto 1/\epsilon(n)$ goes to infinity, and therefore $x$ goes to infinity as $n$ goes to infinity.  Thus we have that the number of elements in any particular bucket (which follows a $B(n,1/k)$ distribution) is larger than $t$ with probability $e^{-x^{2}/2 - \log x - \frac{1}{2} \log \pi + o(1)} \leq e^{-x^{2}/2}$. Putting in $x =\sqrt{2(C+1) \log k (1+ \frac{1}{k-1})}$, we get $x^{2}/2 = (C+1) \log k (1+  1/(k-1) \geq (C+1) \log k$, and thus the probability is at most $e^{-(C+1) \log k} = 1/k^{C+1}$. 
 
Thus the probability that the maximum number of balls in a bin is more than $t$ is bounded (by the union bound) by $k.1/k^{C+1} = 1/k^{C}$, and the lemma is proved.
\end{proof}

Now that we know that with probability as least $1-1/k^{C}$, no bucket has more than $t$ elements, we observe that a bucket serves as a ``mini'' CMS for the elements that hash to it. In other words, let $n(q)$ be the number of elements that hash to the same bucket as $q$ under $h_{0}$. The expected error in the $i$th row of the mini-CMS for $q$ (the entry for which is contained inside the bucket of $q$), is $\mathbb{E}[\texttt{Error}_{i}(q)] = n(q)/(B/r) = n(q) \epsilon k/e$. 

By Markov's inequality $\text{Pr}[ \texttt{Error}_{i}(q) \geq n(q) k\epsilon] \leq 1/e$. 

Let $\alpha = t \epsilon k/e = (n/k + \sqrt{(2(C+1) n \log k)/k}) \epsilon k / e = (n \epsilon/e)(1+\sqrt{(2(C+1)k \log k) / n})$. We now compute the bound on the final error (after taking the min) as follows.

\begin{eqnarray}
Pr(\text{Error}(q) \geq e \alpha) &=& Pr(\text{Error}_{i}(q) \geq e\alpha \ \ \   \forall i \in \{1,\cdots,r\}) \notag \\
&=& Pr(\text{Error}_{i}(q) \geq e\alpha \ \ \   \forall i |  \ \ n(q) \leq t)Pr(n(q) \leq t) \notag \\
&+&  Pr(\text{Error}_{i}(q) \geq e\alpha \ \ \   \forall i | \ \ n(q) \geq t)Pr(n(q) \geq t) \notag \\
&\leq& \left(\frac{1}{e}\right)^{r}1 + 1(1/k^{C})\notag  \\
&=& \delta + 1/k^{C},\notag 
\end{eqnarray}

where the second last equality follows from Markov's inequality on $\text{Error}_{i}(q)$ and Lemma~\ref{binsours}. Finally, by observing that for a fixed $\delta$, $k = O(e/B \epsilon)$, the proof of the theorem is complete.


\section{Evaluation}

In this section, we evaluate our implementation of the \bcms. We compare the
\bcms against the (traditional) \cms.
We evaluate each data structure on two fundamental operations, insertions and
queries. We evaluate queries for  set of elements chosen uniformly at random.

In our evaluation, we address the following questions about how the performance
of \bcms compares to the \cms:

\begin{itemize}

  \item How does the insertion throughput in \bcms compare to \cms on SSD?

  \item How does the query throughput in \bcms compare to \cms on SSD?

  \item How does the hash localization in \bcms affect the overestimates
    compared to the overestimates in \cms?

\end{itemize}

\subsection{Experimental setup}

To answer the above questions, we evaluate the performance of the \bcms and the
\cms on SSD by scaling the sketch out-of-RAM. For SSD benchmarks, we use four
different RAM-size-to-sketch-size ratios, $2$, $4$, $8$, and $16$.  The
RAM-size-to-sketch-size ratio is the ratio of the size of the available RAM and
the size of the sketch on SSD.
We fix the size of the available RAM to be $\approx64$MB and change the size of
the sketch to change the RAM-size-to-sketch-size ratio. The page size in all our
benchmarks was set to $4096$B.
In all the benchmarks, we measure the throughput (operations per second) to
evaluate the insertion and query performance.

To measure the insertion throughput, we first calculate the number of elements we
can insert in the sketch using calculations described in \Cref{calculations}.
During an insert operation, we first generate a $64$-bit integer from a
uniform-random distribution and then add that integer to the sketch. This way,
we do not use any extra memory to store the set of integers to be added to the
sketch.
We then measure the total time taken to insert the given set of elements in the
sketch.
Note that for the \bcms, we make sure to flush all the remaining inserts from
the buffer to the sketch on SSD at the end and include the time to do that in
the total time.

To measure the query throughput, we query for elements drawn from a
uniform-random distribution and measure the throughput. 
The reason for the query benchmark is to simulate a real-world query workload
where some elements may not be present in the sketch and the query will terminate
early thereby requiring fewer I/Os.

For all the query benchmarks, we first perform the insertion benchmark and write
the sketch to SSD. After the insertion benchmark, we flush all caches (page
cache, directory entries, and inodes).
We then map the sketch back into RAM and perform queries on the sketch. This way
we make sure that the sketch is not already cached in kernel caches from the
insertion benchmark.

We compare the overestimates in \bcms and \cms for all the four sketch sizes for
which we perform insertion and query benchmarks. To measure the overestimates,
we first perform the insertion benchmark. However, during the insertion benchmark, we
also store each inserted element in a multiset. Once insertions are done, we
iterate over the multiset and query for each element in the sketch. We then take
the difference of the count returned from the sketch and the actual count of the
element to calculate the overestimate.

For SSD-based experiments, we allocate space for the sketch by \texttt{mmap}-ing it to
a file on SSD. We then control the available RAM to the benchmarking process
using \texttt{cgroups}.
We fix the RAM size for all the experiments to be $\approx67$MB. We then
increase the size of the sketch based on the RAM-size-to-sketch-size ratio of
the particular experiment.
For the \bcms, we use all the available RAM as the buffer. Paging is handled by
the operating system based on the disk accesses. The point of these experiments
is to evaluate the I/O efficiency of sketch operations.

All benchmarks were performed on a $64$-bit Ubuntu 16.04 running Linux kernel
4.4.0-98-generic. The machine has Intel Skylake CPU U (Core(TM) i7-6700HQ CPU @
$2.60$GHz with $4$ cores and $6$MB L$3$ cache) with $32$ GB RAM and $1$TB
Toshiba SSD. 

\begin{table}
  \begin{centering}
    \begin{tabular}{cccc}
    \hline
    \textbf{Size} & \textbf{Width} & \textbf{Depth} & \textbf{\#elements} \\
    \hline
    \hline
    128MB & 3355444 & 5 & 9875188 \\[.05in] 
    256MB & 6710887 & 5 & 19750377 \\[.05in] 
    512MB & 13421773 & 5 & 39500754 \\[.05in] 
    1GB &26843546 & 5 & 79001508 \\[.05in]
    \hline
  \end{tabular}
  \caption{Size, width, and depth of the sketch and the number of elements
    inserted in \cms and \bcms in our benchmarks (insertion, query, and
  overestimate calculation).}
    \label{tab:cms-config}
  \end{centering}
\end{table}

\subsection{Configuring the sketch}
\label{calculations}

In our benchmarks, we take as input $\delta$, overestimate $O$ ($\epsilon n$),
and the size of the sketch to configure the sketch $S$. The depth of the \cms
$D$ is $\lceil\ln{\frac{1}{\delta}}\rceil$. The number of cells $C$ is
$S/CELL\_SIZE$. And width of the \cms is $\lceil e/\epsilon\rceil$.

Given these parameters, we calculate the number of elements $n$ to be inserted in
the sketch as $\frac{C \times O}{D \times e}$. In all our experiments, we set
$\delta$ to $0.01$ and maximum overestimate to $8$ and change the sketch size.
Table \ref{tab:cms-config} shows dimensions of the sketch and number of elements
inserted based on the size of the sketch.

\begin{figure}
\centering
\includegraphics[width=12cm]{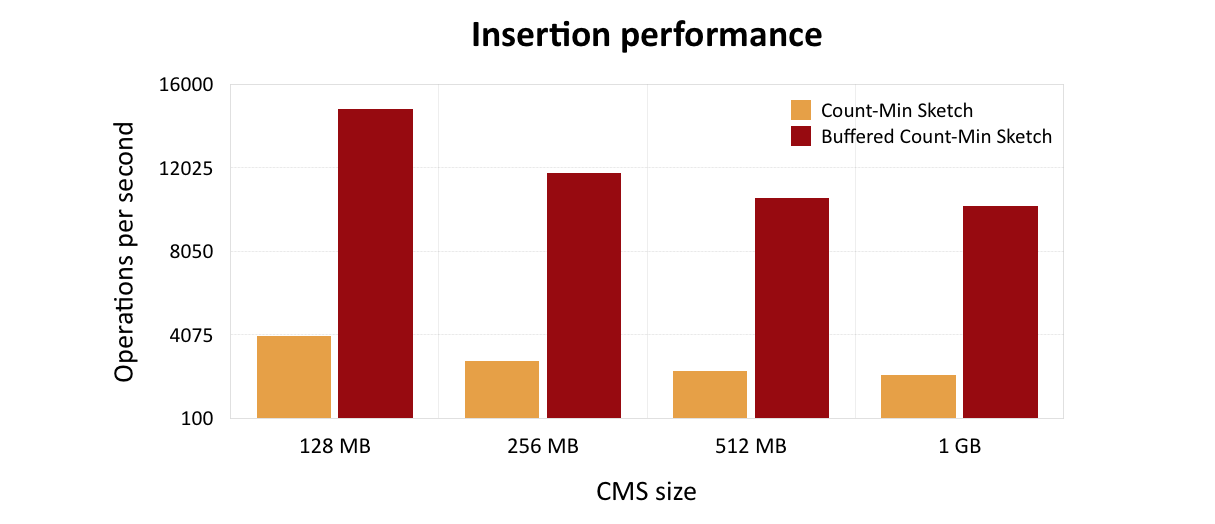}
\caption{Insert throughput of \cms and \bcms with increasing sizes. The available
RAM is fixed to $\approx64$MB. With increasing sketch sizes (on x-axis) the
RAM-size-to-sketch-size is also increasing $2$, $4$, $8$, and $16$. (Higher is
better)}
\label{fig:exp1}
\end{figure}
\begin{figure}
\centering
\includegraphics[width=12cm]{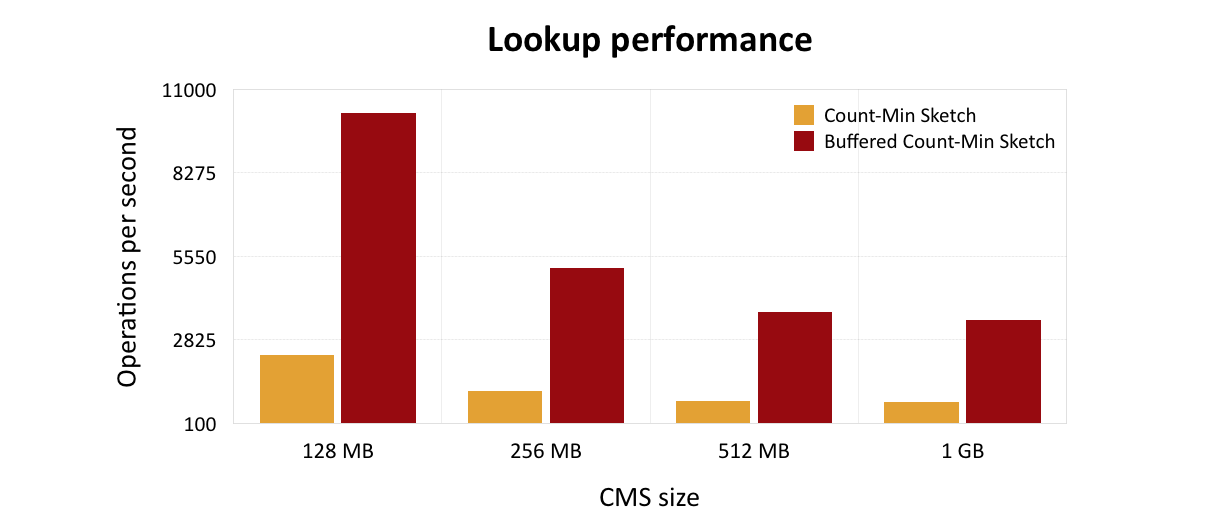}
\caption{Query throughput of \cms and \bcms with increasing sizes. The available
RAM is fixed to $\approx64$MB. With increasing sketch sizes (on x-axis) the
RAM-size-to-sketch-size is also increasing $2$, $4$, $8$, and $16$. (Higher is
better)}
\label{fig:exp2}
\end{figure}
\begin{figure}
\centering
\includegraphics[width=12cm]{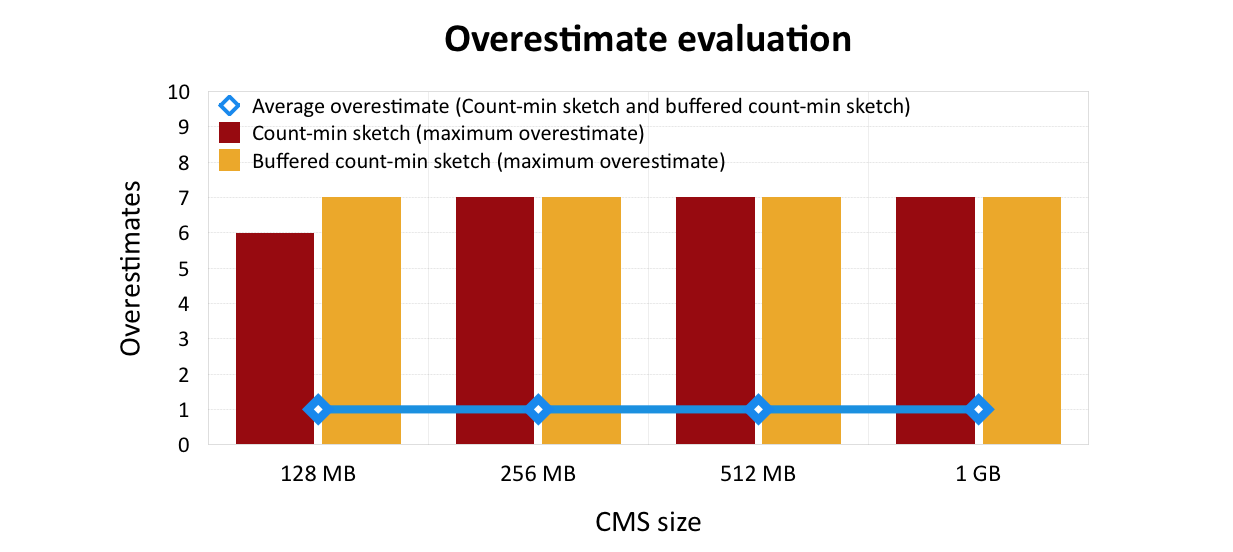}
\caption{Maximum overestimate reported by \cms and \bcms for any inserted
  element for different sketch sizes. The blue line represents the average
overestimate reported by \cms and \bcms for all the inserted elements. The
average overestimate is same for both \cms and \bcms.}
\label{fig:exp3}
\end{figure}

%

\subsection{Insert Performance}

Figure \ref{fig:exp1} shows the insert throughput of \cms and \bcms with
changing RAM-size-to-sketch-size ratios. \bcms is $3.7\times$--$4.7\times$
faster compared to the \cms in terms of insert throughput on SSD.

The \bcms performs less than one I/O per insert operation because all the hashes
for a given element are localized to a single page on SSD. However, in the \cms
the hashes for a given element are spread across the whole sketch. Therefore,
the insert throughput of the \bcms is $3.7\times$ when the sketch is twice the
size of the RAM. And the difference in the throughput increases as the sketch
gets bigger and RAM size stays the same.

\subsection{Query Performance}

Figure \ref{fig:exp2} shows the query throughput of \cms and \bcms with changing
RAM-size-to-sketch-size ratios. \bcms is $\approx4.3\times$ faster compared to
the \cms in terms of query throughput on SSD.

The \bcms performs a single I/O per query operation because all the hashes for a
given element are localized to a single page on SSD. In comparison, \cms may
have to perform as many as $h$ I/Os per query operation, where $h$ is the depth
of the \cms.

\subsection{Overestimates}

In Figure \ref{fig:exp3} we empirically compare overestimates returned by the
\cms and \bcms for all the four sketch sizes for which we performed insert and
query benchmarks. And we found that the average and the maximum overestimate
returned from \cms and \bcms are exactly the same. This shows that empirically
hash localization in \bcms does not have any major effect on the overestimates.

%

%
%

\section{Conclusion}

In this paper we described the design and implementation of the Buffered
count-min sketch, and empirically showed that our implementation achieves
$3.7\times$--$4.7\times$ the speedup on update (insert) and $4.3\times$ speedup
on estimate (lookup) operations. Queries take $1$ I/O, which is optimal in the
worst case if not allowed to buffer. However, we do not know whether the update
time is optimal. To the best of our knowledge, no lower bounds on the update
time of such a data structure are known (the only known upper bounds are on
space, e.g., in \cite{ganguly2008lower}). We leave the question of deriving
update lower bounds and/or a SSD-based data structure with faster update time
for future work.

\bibliography{lipics-v2018-bcms}

\end{document}